\documentclass[12pt]{iopart}
\usepackage{graphicx}
\usepackage{amsthm}
\usepackage{iopams}  
\usepackage{setstack}  
\usepackage{cite}

\newcommand{\real}{\textrm{Re}\,}
\newcommand{\imag}{\textrm{Im}\,}
\newcommand{\parti}[2]{\frac{\partial #1}{\partial #2}}

\newcommand{\intall}{\int_{-\infty}^{\infty}}
\newcommand{\ket}[1]{|#1\rangle}

\newcommand{\bra}[1]{\langle#1|}

\newcommand{\avg}[1]{\langle#1\rangle}
\newcommand{\Avg}[1]{\left\langle#1\right\rangle}

\newcommand{\abs}[1]{\left|#1\right|}
\newcommand{\bk}[1]{\left(#1\right)}
\newcommand{\Bk}[1]{\left[#1\right]}
\newcommand{\BK}[1]{\left\{#1\right\}}

\newcommand{\trace}{\textrm{tr}}

\newtheorem{lemma}{Lemma}
\newtheorem{theorem}{Theorem}
\newtheorem{defin}{Definition}
\newtheorem{cor}{Corollary}
\theoremstyle{remark}

\begin{document}
\title{Quantum metrology with open dynamical systems}

\author{Mankei Tsang}

\ead{eletmk@nus.edu.sg}
\address{Department of Electrical and Computer Engineering,
  National University of Singapore, 4 Engineering Drive 3, Singapore
  117583}

\address{Department of Physics, National University of Singapore,
  2 Science Drive 3, Singapore 117551}











\date{\today}

\begin{abstract}
  This paper studies quantum limits to dynamical sensors in the
  presence of decoherence.  A modified purification approach is used
  to obtain tighter quantum detection and estimation error bounds for
  optical phase sensing and optomechanical force sensing. When optical
  loss is present, these bounds are found to obey shot-noise scalings
  for arbitrary quantum states of light under certain realistic
  conditions, thus ruling out the possibility of asymptotic Heisenberg
  error scalings with respect to the average photon flux under those
  conditions. The proposed bounds are expected to be approachable
  using current quantum optics technology.
\end{abstract}

\maketitle
\section{Introduction}
The laws of quantum mechanics impose fundamental limitations to the
accuracy of measurements, and a fundamental question in quantum
measurement theory is how such limitations affect precision sensing
applications, such as gravitational-wave detection, optical
interferometry, and atomic magnetometry and gyroscopy
\cite{braginsky,glm_science}. With the rapid recent advance in quantum
optomechanics \cite{kippenberg,aspelmeyer,ligo2011,schnabel,chen2013}
and atomic \cite{chu,budker} technologies, quantum sensing limits have
received renewed interest and are expected to play a key role in
future precision measurement applications.

Many realistic sensors, such as gravitational-wave detectors, perform
continuous measurements of time-varying signals (commonly called
waveforms). For such sensors, a quantum Cram\'er-Rao bound (QCRB) for
waveform estimation \cite{twc} and a quantum fidelity bound for
waveform detection \cite{tsang_nair} have recently been proved,
generalizing earlier seminal results by Helstrom
\cite{helstrom}. These bounds are not expected to be tight when
decoherence is significant, however, as \cite{twc,tsang_nair}
use a purification approach that does not account for the
inaccessibility of the environment. Given the ubiquity of decoherence
in quantum experiments, the relevance of the bounds to practical
situations may be questioned.

One way to account for decoherence is to employ the concepts of mixed
states, effects, and operations \cite{kraus}. Such an approach has
been successful in the study of single-parameter estimation problems
\cite{fujiwara,kolodynski,knysh,demkowicz,escher,escher_bjp,escher_prl,kolodynski2013},
but becomes intractable for nontrivial quantum dynamics. To retain the
convenience of a pure Hilbert space, here I extend a modified
purification approach proposed in
\cite{fujiwara,escher,escher_bjp,escher_prl} and apply it to
more general open-system detection and estimation problems beyond the
paradigm of single-parameter estimation considered by previous work
\cite{fujiwara,kolodynski,knysh,demkowicz,escher,escher_bjp,escher_prl,kolodynski2013}.
In particular, I show that
\begin{enumerate}
\item For optical phase detection with loss and vacuum noise, the
  errors obey lower bounds that scale with the average photon number
  akin to reduced shot-noise limits, provided that the phase shift or
  the quantum efficiency is small enough (the precise conditions will
  be given later).  This rules out Heisenberg scaling of the
  detectable phase shift \cite{ou,paris1997} in the high-number limit
  under such conditions, as well as any significant enhancement of the
  error exponent by quantum illumination \cite{lloyd,tan,pirandola} in
  the low-efficiency limit with vacuum noise. Similar results exist
  when the phase is a waveform.
 
\item The mean-square error for lossy optical phase waveform
  estimation also observes a limit with shot-noise scaling, which
  generalizes the single-parameter results in
  \cite{kolodynski,knysh,demkowicz,escher,escher_bjp,escher_prl}
  and rules out the kind of quantum-enhanced scalings suggested by
  \cite{berry2002,berry2006,tsl2008,tsl2009} in the high-flux
  limit.
\item A quantum model of optomechanical force sensing can be
  transformed to an optical phase sensing problem with classical phase
  shift, such that a unified formalism can treat both problems and
  produce tighter bounds than the results in
  \cite{twc,tsang_nair}.
\end{enumerate}
These results not only provide more general and realistic quantum
limits that can be approached using current quantum optics technology
\cite{cook,wittmann,tsujino,wheatley,yonezawa,iwasawa}, but may also be
relevant to more general studies of quantum metrology and quantum
information, such as quantum speed limits \cite{taddei,delcampo} and
Loschmidt echo \cite{gorin}.

\section{\label{purify}The modified purification approach}
Let $x$ be the vector of unknown parameters to be estimated and $y$ be
the observation. Within the purification approach
\cite{twc,tsang_nair}, the dynamics of a quantum sensor is modeled by
unitary evolution ($U_x$ as a function of $x$) of an initial pure
density operator $\rho = \ket{\Psi}\bra{\Psi}$, and measurements are
modeled by a final-time positive operator-valued measure (POVM) $E(y)$
using the principle of deferred measurement \cite{nielsen}.  The
likelihood function becomes 
\begin{eqnarray}
P(y|x) = \trace\Bk{E(y) U_x \rho U_x^\dagger}.
\end{eqnarray}
For continous-time problems, discrete time is first assumed and the
continuous limit is taken at the end of
calculations. \cite{twc,tsang_nair} derive quantum bounds by
considering the density operator $U_x \rho U_x^\dagger$.

\begin{figure}[htbp]
\centerline{\includegraphics[width=0.6\textwidth]{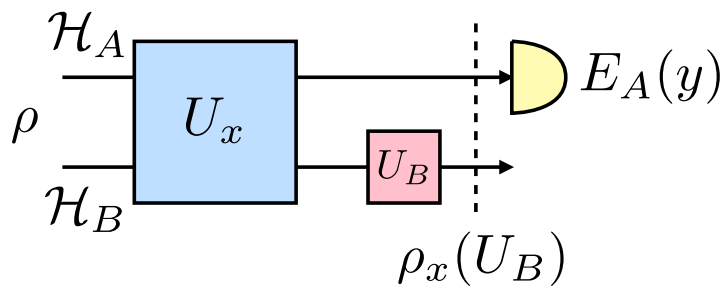}}
\caption{Quantum circuit diagram \cite{nielsen} for the modified
  purification approach.}
\label{naimark}
\end{figure}

Suppose that the Hilbert space ($\mathcal H = \mathcal H_A\otimes
\mathcal H_B$) is divided into an accessible part ($\mathcal H_A$) and
an inaccessible part ($\mathcal H_B$). The POVM should now be written
as $E(y) = E_A(y)\otimes 1_B$, where $E_A(y)$ is a POVM on $\mathcal
H_A$ and $1_B$ is the identity operator on $\mathcal H_B$, which
accounts for the fact that $\mathcal H_B$ cannot be measured.  The key
to the modified purification approach, as illustrated by
Fig.~\ref{naimark}, is to recognize that the likelihood function is
unchanged if any arbitrary $x$-dependent unitary $U_B$ on $\mathcal
H_B$ is applied before the POVM:
\begin{eqnarray}
P(y|x) = \trace\BK{\Bk{E_A(y)\otimes 1_B } \rho_x(U_B)},
\label{rhoxUB}
\end{eqnarray}
where
\begin{eqnarray}
\rho_x(U_B) \equiv (1_A\otimes U_B^\dagger) U_x \rho U_x^\dagger (1_A\otimes U_B)
\label{rhox}
\end{eqnarray}
is a purification of $\trace_B(U_x\rho U_x^\dagger)$, such that
$\trace_B \rho_x = \trace_B(U_x\rho U_x^\dagger)$.  Judicious choices
of $U_B$ can result in tighter quantum bounds as a function of
$\rho_x(U_B)$ \cite{fujiwara,escher,escher_bjp,escher_prl}.

First, suppose that $x^{(0)}$ and $x^{(1)}$ are the two
hypotheses for $x$ and $\tilde x(y)$ is the estimate. The following
theorems are applications of the modified purification and Helstrom's
bounds for pure states \cite{helstrom}:
\begin{theorem}[Fidelity bound, Neyman-Pearson criterion]
\label{np}
For any POVM measurement $E_A(y)$ of $\trace_B \rho_x$ in $\mathcal
H_A$, the miss probability, defined as
\begin{eqnarray}
P_{01}\equiv \int_{\tilde x(y) = x^{(0)}} dy P(y|x^{(1)}),
\end{eqnarray}
given a constraint on the false-alarm probability 
\begin{eqnarray}
P_{10} \equiv
\int_{\tilde x(y) = x^{(1)}}dy P(y|x^{(0)})\le \alpha,
\end{eqnarray}
satisfies
\begin{eqnarray}
P_{01} \ge \beta(\alpha,F)
\equiv 
\bigg\{\begin{array}{ll}
1-\Bk{\sqrt{\alpha F} + \sqrt{(1-\alpha)(1-F)}}^2,
&\alpha < F,
\\
0,& \alpha \ge F,
\end{array}
\label{beta}
\end{eqnarray}
where $F$ is the fidelity between the following pure states in
$\mathcal H_A\otimes \mathcal H_B$:
\begin{eqnarray}
\rho_0 \equiv U_B^\dagger U_0\ket{\Psi}\bra{\Psi}U_0^\dagger U_B,
\\
\rho_1 \equiv  U_1\ket{\Psi}\bra{\Psi}U_1^\dagger,
\\
F(\rho_0,\rho_1) \equiv 
\abs{\bra{\Psi} U_{1}^\dagger U_B^\dagger U_0 \ket{\Psi}}^2,
\label{pure_fidelity}
\end{eqnarray}
$\rho_m \equiv \rho_{x^{(m)}}$, $U_m \equiv U_{x^{(m)}}$, and
$1_A\otimes U_B$, abbreviated as $U_B$, is an arbitrary unitary on
$\mathcal H_B$.
\end{theorem}
\begin{theorem}[Fidelity bound, Bayes criterion]
\label{Pe}
The average error probability $P_e\equiv P_{10}P_0+P_{01}P_1$ with
prior probabilities $P_0$ and $P_1=1-P_0$ satisfies
\begin{eqnarray}
P_e \ge \frac{1}{2}\bk{1-\sqrt{1-4P_0P_1 F}} \ge P_0P_1 F.
\label{helstrom}
\end{eqnarray}
\end{theorem}
\begin{proof}[Proof of Theorem \ref{np} and \ref{Pe}]
  \cite{helstrom} shows that the bounds with the likelihood
  function $P(y|x)=\trace[E(y)\rho_x]$ are valid for any POVM $E(y)$
  on $\mathcal H_A\otimes\mathcal H_B$, so they must also be valid
  with $P(y|x) = \trace_A [E_A(y)\trace_B\rho_x] =
  \trace\{[E_A(y)\otimes 1_B]\rho_x\}$ for any POVM $E_A(y)$ on
  $\mathcal H_A$.
\end{proof}

  Since the lower bounds are valid for any $U_B$, $U_B$ should be
  chosen to increase $F$ and tighten the bounds. The maximum $F$
  becomes the Uhlmann fidelity between mixed states $\trace_B \rho_0$
  and $\trace_B \rho_1$ \cite{nielsen,wilde}. The bound on $P_e$
  obtained using this method is thus weaker than the Helstrom bound
  for the mixed states \cite{helstrom}, although it can be shown that
  the error exponent for the Uhlmann-fidelity bound is within 3dB of
  the optimal value \cite{audenaert2008}. 


Next, consider the estimation of continuous parameters $x$ with prior
distribution $P(x)$. A lower error bound is given by the following:
\begin{theorem}[Bayesian quantum Cram\'er-Rao bound]
  The error covariance matrix 
\begin{eqnarray}
\Sigma\equiv \mathbb E \bk{\tilde
    x-x}\bk{\tilde x-x}^\top = \int dy dx P(y|x)P(x)\bk{\tilde
    x-x}\bk{\tilde x-x}^\top
\end{eqnarray}
satisfies a matrix inequality given by
\begin{eqnarray}
\Sigma \ge \bk{J^{(Q)}+J^{(C)}}^{-1},
\label{matrix_qcrb}
\end{eqnarray}
where
\begin{eqnarray}
J_{jk}^{(Q)} &= -2\int dx P(x) 
 \parti{^2F(\rho_{x},\rho_{x'})}{x_j'\partial x_k'} \Bigg|_{x'=x},
\label{J_F}
\\
J_{jk}^{(C)} &= \int dx P(x) \parti{\ln P(x)}{x_j}\parti{\ln P(x)}{x_k}.
\end{eqnarray}
\end{theorem}
\begin{proof}
  See \cite{twc} for a proof of (\ref{matrix_qcrb}). To
  relate $J^{(Q)}$ and $F$ as in (\ref{J_F}), first note the
  identity that relates the Bures distance $D_B^2(\rho_x,\rho_{x+dx})$
  between two density matrices separated by an infinitesimal 
  parameter change to the quantum Fisher
  information (QFI) matrix $\mathcal J^{(Q)}(x)$ \cite{hayashi,paris}:
\begin{eqnarray}
D_B^2(\rho_x,\rho_{x+dx}) = 2\Bk{1-\sqrt{F(\rho_x,\rho_{x+dx})}} = \frac{1}{4} 
\sum_{j,k} \mathcal J_{jk}^{(Q)}(x) dx_j dx_k.
\end{eqnarray}
This allows one to write the fidelity as
\begin{eqnarray}
F(\rho_x,\rho_{x+dx}) = 1-\frac{1}{4} \sum_{j,k} \mathcal J_{jk}^{(Q)}(x) dx_j dx_k,
\end{eqnarray}
and $\mathcal J^{(Q)}(x)$ as
\begin{eqnarray}
\mathcal J_{jk}^{(Q)}(x) = -2\parti{^2F(\rho_{x},\rho_{x'})}{x_j'\partial x_k'} \Bigg|_{x'=x}.
\end{eqnarray}
According to \cite{twc}, the Bayes QFI $J^{(Q)}$ is the average
of $\mathcal J^{(Q)}(x)$ over the prior probability distribution
$P(x)$. (\ref{J_F}) then follows.
\end{proof}
Here I focus on the Bayes version of the QCRB because the inclusion of
prior information is crucial in waveform estimation
\cite{twc,vantrees}. The Bayes bound also has the advantage of being
applicable to both biased and unbiased estimates
\cite{twc,vantrees}. $U_B$ should again be chosen to reduce
$J^{(Q)}(U_B)$ and thus tighten the QCRB.
  

An alternative to the QCRB is a multiparameter form of the quantum
Ziv-Zakai bound \cite{qzzb,bell}, which can also be expressed in
terms of the fidelity, but that option is beyond the scope of this
paper.

\section{Lossy optical phase detection}

\begin{figure}[htbp]
\centerline{\includegraphics[width=0.6\textwidth]{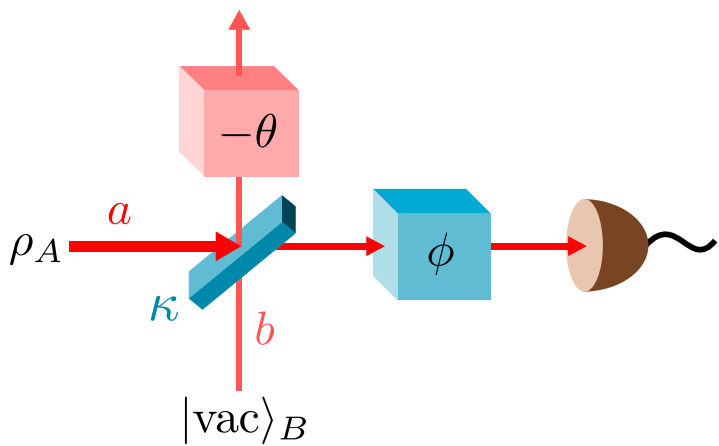}}
\caption{A model of the lossy optical phase sensing problem.}
\label{lossy_static_phase}
\end{figure}

To introduce a new technique of bounding the fidelity, I first
consider the simplest setting of an optical phase detection problem,
where one optical mode is used to detect a phase shift
\cite{ou,paris1997}, as depicted in Fig.~\ref{lossy_static_phase}.
Let $\phi$ the phase shift between the two hypotheses, $U_0 = U_A
U_{AB}$, and $U_1 = U_{AB}$, where
\begin{eqnarray}
U_A &= \exp\bk{i\phi n},
\qquad
n \equiv a^\dagger a,
\\
U_{AB} &= \exp\Bk{i \kappa\bk{a^\dagger b + a b^\dagger}}.
\end{eqnarray}
$a$ and $b$ are annihilation operators for two different modes that
satisfy commutation relations $[a,a^\dagger] = [b,b^\dagger] = 1$, $n$
is the photon-number operator for the $A$ mode, $U_{AB}$ models loss
as a beam-splitter coupling with another optical mode $B$ in vacuum
state $\ket{0}_B$ before the phase modulation, such that $\ket{\Psi} =
\ket{\psi}_A\otimes\ket{0}_B$. $U_{AB}$ can also account for loss
after the modulation, as shown in \ref{order}. The fidelity becomes
\begin{eqnarray}
F = \abs{\bra{\Psi}U_{AB}^\dagger U_B^\dagger U_A U_{AB}\ket{\Psi}}^2.
\label{F}
\end{eqnarray}
Choosing
\begin{eqnarray}
U_B = \exp\bk{i \theta b^\dagger b},
\end{eqnarray}
where $\theta$ is a free parameter to be specified later, one can
simplify (\ref{F}) using the $SU(2)$ disentangling theorem
\cite{gilmore}, as shown in \ref{su2}. The result is
\begin{eqnarray}
F &= \abs{\bra{\psi} z^{n}\ket{\psi}}^2 ,
\label{ztransform}
\\
z &\equiv \eta e^{i\phi} + (1-\eta) e^{-i\theta},
\end{eqnarray}
with 
\begin{eqnarray}
\eta \equiv \cos^2\kappa
\end{eqnarray}
defined as the quantum efficiency. For example, as shown in \ref{su2},
the fidelity for a coherent state is
\begin{eqnarray}
F_{\textrm{coh}} &=\exp\bk{- 4\eta \avg{n} \sin^2\frac{\phi}{2}},
\label{Fcoh}
\end{eqnarray}
where $\avg{O}\equiv \bra{\psi}O\ket{\psi}$.  Since $-\ln
F_{\textrm{coh}}$ depends linearly on the average photon number
$\avg{n}$, I shall define the linear scaling of $-\ln F$ with respect
to $\avg{n}$ as the shot-noise scaling for the fidelity.  Measurements
that can saturate the Bayes error bound in Theorem~\ref{Pe} for
coherent states are known \cite{helstrom,cook,wittmann,tsujino}.

To bound $F$ in general, Jensen's inequality can be used if $z$ is
real and positive.  The following lemma provides the necessary and
sufficient condition:
\begin{lemma}
\label{conds}
There exists a $\theta$ such that $z\equiv \eta
e^{i\phi}+(1-\eta)e^{-i\theta}$ is real and positive if and only
if one of the following conditions is satisfied:
\begin{eqnarray}
{\rm(I)}&:\ \eta < \frac{1}{2},\label{cond1}\\
{\rm(II)}&:\ \eta \ge \frac{1}{2}
{\rm\ and\ }
|\sin\phi| \le \frac{1-\eta}{\eta}
{\rm\ and\ }
\cos\phi > 0.
\label{cond2}
\end{eqnarray}
\end{lemma}
\begin{proof}
  Consider the circle traced by $z(\theta)$ centered at $\eta
  e^{i\phi}$ with radius $1-\eta$ on the complex plane. $z = |z| > 0$
  for some $\theta$ is equivalent to the condition that the circle
  intersects the positive real axis, for which the necessary and
  sufficient condition is given by one of (\ref{cond1}) (the
  circle encloses the origin for any $\phi$ and thus always intersects
  the axis) and (\ref{cond2}) (the circle intersects the axis for
  some $\phi$ on the right-hand plane only).
\end{proof}

(\ref{cond2}) holds when $\phi$ is sufficiently small.  For
example, $M\sim 100$, $|q|\sim 10^{-19}~$m, $2\pi/k\sim 1~\mu$m, and
$(1-\eta)/\eta \sim 10^{-2}$ for LIGO \cite{klmtv}, leading to $|\phi|
\sim 2Mk|q| \sim 10^{-10}$, and (\ref{cond2}) is easily satisfied.

The following theorem is a key technical result of this paper:
\begin{theorem}
\label{Fzbound}
  If (\ref{cond1}) or (\ref{cond2}) is satisfied,
\begin{eqnarray}
F &\ge F_z \equiv z^{2\avg{n}},
\label{Fzdef}
\end{eqnarray}
where
\begin{eqnarray}
z \equiv \eta e^{i\phi} + (1-\eta)e^{-i\theta},
\end{eqnarray}
and $\theta$ is chosen to make $z$ real and positive.
\end{theorem}
\begin{proof}
  With $z = |z| > 0$ under the condition in Lemma~\ref{conds} and
  writing $\ket{\psi}$ as a superposition of eigenstates of $n$, one
  can apply Jensen's inequality and obtain $\avg{z^n} \ge
  z^{\avg{n}}$.  (\ref{Fzdef}) then follows from
  (\ref{ztransform}).
\end{proof}
Compared with the coherent-state value given by (\ref{Fcoh}),
$F_z$ has the same shot-noise scaling with respect to the average
photon number $\avg{n}$, as both $-\ln F_{\textrm{coh}}$ and $-\ln
F_z$ scale linearly with $\avg{n}$.  Since error-free detection with
$F = 0$ is possible with pure states \cite{dariano}, this
shot-noise-scaling bound is a very strong result. It should also have
implications for M-ary phase discrimination in general
\cite{nair,nair_guha}.

The following corollaries are some analytic consequences of
Theorem~\ref{Fzbound} that exemplify its tightness:
\begin{cor}
\label{smallphase}
If $|\phi| \ll 1$ and $|\eta\phi/(1-\eta)| \ll 1$,
\begin{eqnarray}
F_z 
&\approx \exp\Bk{-\frac{\eta}{1-\eta}\avg{n}\phi^2}.
\label{Fzsmallphase}
\end{eqnarray}
\end{cor}
\begin{proof}
  Let $\theta_0\equiv\eta\phi/(1-\eta)$.  Since $\imag z =
  \eta\sin\phi-(1-\eta)\sin \theta = 0$, $\theta = \theta_0[1+
  O(\phi^2)+ O(\theta_0^2)]$, $z^2
  =1-4\eta(1-\eta)\sin^2[(\phi+\theta)/2]= 1-\theta_0\phi [1+
  O(\phi^2) + O(\theta_0^2)]$, and $\ln z^2 =
  -\theta_0\phi[1+ O(\theta_0\phi) + O(\phi^2) +
  O(\theta_0^2)]$, which determines $F_z$ in
  (\ref{Fzdef}) and leads to (\ref{Fzsmallphase}).
\end{proof}
(\ref{Fzsmallphase}) differs from $F_{\textrm{coh}} \approx
\exp(-\eta\avg{n}\phi^2)$ by just a constant factor of $1/(1-\eta)$ in
the exponent. This $1/(1-\eta)$ enhancement factor is the same as the
maximum QFI enhancement factor in lossy static-phase estimation
\cite{kolodynski,knysh,demkowicz,escher,escher_bjp,escher_prl}.

To obtain another measure of detection error, I formalize the concept
of detectable phase shift \cite{ou,paris1997} as follows:
\begin{defin}[Detectable phase shift]
  A detectable phase shift $\phi'$ given acceptable error
  probabilities $\alpha'$ and $\beta'$ is a $\phi$ that makes
  $P_{10}\le \alpha'$ and $P_{01} \le \beta'$.
\end{defin}
\begin{cor}
\label{detectable_bound}
  Assuming that (\ref{Fzsmallphase}) is an equality,
\begin{eqnarray}
   \phi'^2 &\ge
  \frac{1-\eta}{\eta\avg{n}}(-\ln F'),
\label{detectable}
\end{eqnarray}
where 
\begin{eqnarray}
F'\equiv\max_{\beta' \ge \beta(\alpha',F)} F
\label{Fprime}
\end{eqnarray}
and $\beta$ is defined in (\ref{beta}).
\end{cor}
\begin{proof}
  Any achievable $(P_{10},P_{01})$ must lie above the convex curve
  $P_{01}=\beta(P_{10},F)$.  This means that for $P_{10}\le \alpha'$
  and $P_{01} \le \beta'$, $\beta' \ge \beta(\alpha',F)$ must hold,
  and hence $F \le \max_{\beta'\ge \beta(\alpha',F)} F\equiv
  F'(\alpha',\beta')$. Since $F \ge F_z$, $F_z \le F'$ must hold for
  the constraints on $(P_{10},P_{01})$ to be
  possible. (\ref{detectable}) then follows from
  (\ref{Fzsmallphase}).
\end{proof}
The lower bound in (\ref{detectable}) is lower than the shot-noise
limit by a constant factor of $1-\eta$ only, ruling out the kind of
Heisenberg scaling suggested by \cite{ou,paris1997} for lossy
weak phase detection in the $\avg{n}\to\infty$ limit.
\begin{cor}
\label{loweta}
If $\eta \ll 1$, $F_z \approx F_{{\rm coh}}$.
\end{cor}
\begin{proof}
  Since $\imag z = 0$, $\theta = \sin^{-1}[\eta\sin\phi/(1-\eta)]
  = O(\eta)$, $z = \real z = \eta\cos\phi + (1-\eta) \cos\theta =
  1+\eta\cos\phi -\eta +O(\eta^2)$, and $\ln z^2 =
  -4\eta\sin^2(\phi/2) + O(\eta^2)$, which leads to $-\ln F_z =
  -[1+O(\eta)]\ln F_{\textrm{coh}}$.
\end{proof}
Corollary~\ref{loweta} proves that, analogous to the case of target
detection \cite{nair2011}, the coherent state is near-optimal for any
phase detection problem in the low-efficiency limit with vacuum noise,
ruling out any significant enhancement of the error exponent by
quantum illumination \cite{lloyd,tan,pirandola}. It remains an open
question whether quantum illumination is useful for high-thermal-noise
low-efficiency phase detection, as \cite{lloyd,tan} show that
quantum illumination is useful for low-efficiency target detection
only when the thermal noise is high.

\section{\label{waveform_detect}Waveform detection}
\begin{figure}[htbp]
\centerline{\includegraphics[width=0.6\textwidth]{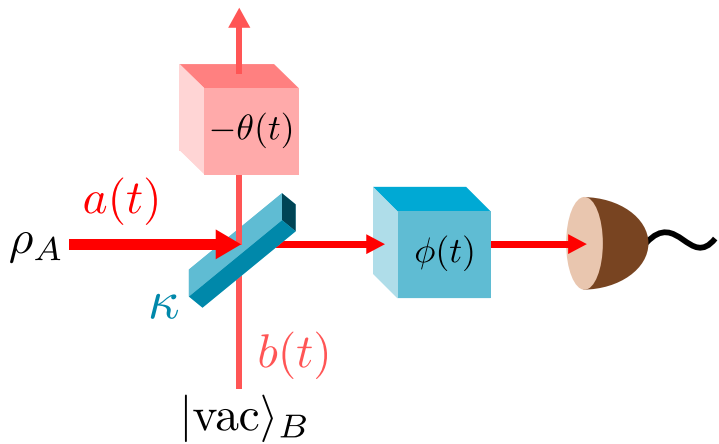}}
\caption{A model of the lossy optical phase waveform sensing problem.
}
\label{lossy_phase}
\end{figure}

I now turn to the problem of waveform detection, as depicted in
Fig.~\ref{lossy_phase}. The results are all natural generalizations of
the single-mode case. Let $\phi(t)$ be the time-varying phase shift
between the two hypotheses, $U_0 = U_A U_{AB}$, and $U_1 = U_{AB}$,
where
\begin{eqnarray}
U_A &= \exp\Bk{i\int dt  \phi(t) I(t)},
\qquad
I(t) \equiv a^\dagger(t)a(t),
\\
U_{AB} &= \exp\BK{i \kappa\int dt
\Bk{a^\dagger(t) b(t) + a(t) b^\dagger(t)}},
\qquad \int dt \equiv  \int_{t_0}^{t_f} dt.
\label{UAB}
\end{eqnarray}
$a(t)$ and $b(t)$ are now annihilation operators for one-dimensional
optical fields with commutation relations $[a(t),a^\dagger(t')] =
[b(t),b^\dagger(t')]= \delta(t-t')$ and $I(t)$ is the photon flux
\cite{gardiner_zoller}. Choosing
\begin{eqnarray}
U_B = \exp\Bk{i \int dt \theta(t) b^\dagger(t) b(t)},
\label{UB}
\end{eqnarray}
where $\theta(t)$ is a free function to be specified later, the
fidelity can be computed by a discrete-time approach. This is done by
deriving the fidelity for the multimode case, writing $a(t_j) =
\sqrt{\delta t}a_j$ and $b(t_j) = \sqrt{\delta t}b_j$ in terms of the
discrete-mode operators $a_j$ and $b_j$ and time $t_j = t_0+j\delta
t$, and taking the $\delta t\to 0$ continuous-time limit at the end of
the calculations. The result is
\begin{eqnarray}
F[\phi(t)] = \abs{\bra{\Psi}U_{AB}^\dagger U_B^\dagger U_A U_{AB}\ket{\Psi}}^2
=\abs{\bra{\psi} 
\exp\Bk{\int dt I(t)\ln z(t)}\ket{\psi}}^2 ,
\label{ztransform2}
\\
z(t) \equiv \eta e^{i\phi(t)} + (1-\eta) e^{-i\theta(t)}.
\end{eqnarray}
For example, the coherent-state value is
\begin{eqnarray}
F_{\textrm{coh}}[\phi(t)] &=\exp\Bk{- 4\eta \int dt \avg{I(t)} \sin^2\frac{\phi(t)}{2}}.
\label{Fcoh2}
\end{eqnarray}
$-\ln F_{\textrm{coh}}[\phi(t)]$ scales linearly with $\avg{I(t)}$,
and this linear scaling shall be defined as the shot-noise scaling for
waveform detection. 

Jensen's inequality can again be used to bound the fidelity in
(\ref{ztransform2}) if $z(t)$ is real and positive. The
generalizations of Theorem~\ref{Fzbound} and Corollaries
\ref{smallphase}, \ref{detectable_bound}, and \ref{loweta} for
waveform detection are listed below (with the proofs omitted because
they are straightforward generalizations of the ones in the
single-mode case):
\begin{cor}
If (\ref{cond1}) or (\ref{cond2}) is satisfied for all
  $\phi(t)$,
\begin{eqnarray}
F[\phi(t)] \ge F_z[\phi(t)]
\equiv \exp \Bk{\int dt \Avg{I(t)}\ln z^2(t)},
\label{Fzdef2}
\end{eqnarray}
where
\begin{eqnarray}
z(t) \equiv \eta e^{i\phi(t)} + (1-\eta)e^{-i\theta(t)},
\end{eqnarray}
and $\theta(t)$ is chosen to make $z(t)$ real and positive.
\end{cor}
This bound is also a shot-noise-scaling bound, as $-\ln F_z[\phi(t)]$
scales linearly with $\avg{I(t)}$.

\begin{cor}
\label{smallphase2}
If $|\phi(t)| \ll 1$ and $|\eta\phi(t)/(1-\eta)| \ll 1$,
\begin{eqnarray}
F_z[\phi(t)]
&\approx \exp\Bk{-\frac{\eta}{1-\eta}\int dt \avg{I(t)}\phi^2(t)}.
\label{Fzsmallphase2}
\end{eqnarray}
\end{cor}
The enhancement factor $1/(1-\eta)$ in the exponent is the same as
that in the single-mode case in (\ref{Fzsmallphase}).

\begin{cor}
Assuming that (\ref{Fzsmallphase2}) is an equality, a detectable
phase shift $\phi'(t)$ satisfies
\begin{eqnarray}
  \int dt \avg{I(t)} \phi'^2(t) &\ge
  \frac{1-\eta}{\eta}(-\ln F'),
\label{detectable2}
\end{eqnarray}
where $F'$ is defined in (\ref{Fprime}).
\end{cor}
(\ref{detectable2}) is a more general form of the shot-noise
limit on a detectable phae shift. For example, if $\avg{I(t)}$ is
constant, the time-averaged energy of the detectable phase shift
satisfies
\begin{eqnarray}
\frac{1}{t_f-t_0}\int dt \phi'^2(t) &\ge
\frac{1-\eta}{\eta N}(-\ln F'),
\end{eqnarray}
with $N \equiv (t_f-t_0) \avg{I}$.

\begin{cor}
\label{loweta2}
If $\eta \ll 1$, $F_z[\phi(t)] \approx F_{{\rm coh}}[\phi(t)]$.
\end{cor}
This means that, similar to the single-mode case, the coherent state
is near-optimal for phase waveform detection in the low-efficiency
limit with vacuum noise, and no significant quantum enhancement is
possible even when multiple modes are available.

The derivations so far assume that the waveform $\phi(t)$ is known 
exactly. Error bounds for stochastic waveform detection
can also be obtained by averaging $F[\phi(t)]$ over the prior
statistics of $\phi(t)$, as shown in \cite{tsang_nair},
and analytically tractable bounds can be obtained if the prior
statistics are Gaussian and $F[\phi(t)]$ is also Gaussian,
such as the one given by (\ref{Fzsmallphase2}).

\section{\label{estimation}Waveform estimation}
Consider now the waveform estimation problem, using the same model
shown in Fig.~\ref{lossy_phase}. Unlike the previous section, where
the free function $\theta(t)$ is chosen to be an instantaneous
function of $\phi(t)$, here I assume $\theta(t) = \int d\tau
\lambda(t-\tau)\phi(\tau)$, as this can result in an even tighter but
still analytically tractable bound.  The QFI
$J^{(Q)}(t_j,t_k)\equiv\lim_{\delta t\to 0} J_{jk}^{(Q)}/\delta t^2$
for estimating $\phi(t)$ is calculated in \ref{qfi} and given by
\begin{eqnarray}
J_\phi^{(Q)}(t,t') &= 
4\int d\tau d\tau' 
\Big[r_1(t,\tau)r_1(t',\tau')
\Avg{\Delta I(\tau)\Delta I(\tau')}
\nonumber\\&\quad
+ r_2(t,\tau)r_2(t',\tau')\Avg{I(\tau)} \delta(\tau-\tau')
\Big],
\nonumber\\
r_1(t,\tau) &\equiv \eta \delta(t-\tau)-(1-\eta)\lambda(\tau-t),
\nonumber\\
r_2(t,\tau) &\equiv \sqrt{\eta(1-\eta)}\Bk{\delta(t-\tau)+
\lambda(\tau-t)},
\nonumber\\
\Delta I(t) &\equiv I(t)-\avg{I(t)}.
\label{JQphi}
\end{eqnarray}
If the light source has
stationary statistics, $\avg{I(t)}$ is constant, $\avg{\Delta
  I(t)\Delta I(t')}$ depends on $t-t'$ only, and a power spectral
density $S_{\Delta I}(\omega)$ can be defined by 
\begin{eqnarray}
\avg{\Delta I(t)\Delta I(t')} = \intall \frac{d\omega}{2\pi}
S_{\Delta I}(\omega)\exp[i\omega(t-t')].
\end{eqnarray}
A spectral form of $J_\phi^{(Q)}$ in the limit
$(t_0,t_f)\to(-\infty,\infty)$ is then 
\begin{eqnarray}
J_\phi^{(Q)}(\omega) =
4\Bk{|\eta-(1-\eta)\lambda(\omega)|^2 S_{\Delta I}(\omega)+
|1+\lambda(\omega)|^2\eta(1-\eta)\avg{I}},
\\
\lambda(\omega)\equiv
\int dt\lambda(t)\exp(i\omega t).
\end{eqnarray}
The minimum QFI becomes
\begin{eqnarray}
\min_\lambda J_\phi^{(Q)}(\omega) = 
4\Bk{\frac{1}{S_{\Delta I}(\omega)} + \frac{1-\eta}{\eta \Avg{I}}}^{-1}.
\label{min_J}
\end{eqnarray}
This is a generalization of earlier results for lossy static-phase
estimation in
\cite{kolodynski,knysh,demkowicz,escher,escher_bjp,escher_prl}.
Assuming further that $\phi(t)$ is a linear functional of the waveform
of interest $x(t)$,
\begin{eqnarray}
\phi(t)= \int dt' g(t-t')x(t'),
\label{functional}
\end{eqnarray}
a QCRB is then \cite{twc}
\begin{eqnarray}
\mathbb E [\tilde x(t)-x(t)]^2
\ge \intall \frac{d\omega}{2\pi}
\frac{1}{|g(\omega)|^2\min_\lambda J_\phi^{(Q)}(\omega) + J_x^{(C)}(\omega)},
\end{eqnarray}
where $g(\omega) \equiv \int dt g(t)\exp(i\omega t)$ and
$J_x^{(C)}(\omega)$ is the prior information in spectral form. For a
coherent state,
\begin{eqnarray}
\avg{\Delta I(t)\Delta I(t')}_{\textrm{coh}} = \avg{I}\delta(t-t'),
\\
S_{\Delta I,\textrm{coh}}(\omega) = \avg{I},
\\
\min_\lambda J_{\phi,\textrm{coh}}^{(Q)}(\omega) = 4\eta \avg{I}.
\end{eqnarray}
Compared with the coherent-state value,
the QFI for any state is limited by the same shot-noise scaling:
\begin{eqnarray}
\min_\lambda J_\phi^{(Q)}(\omega) \le \frac{4\eta \avg{I}}{1-\eta}.
\end{eqnarray}
This rules out the kind of quantum-enhanced scalings suggested by
\cite{berry2002,berry2006,tsl2008,tsl2009} in the high-flux
limit when loss is present.

\section{\label{sec_optomech}Optomechanical force sensing}
\begin{figure}[htbp]
\centerline{\includegraphics[width=0.6\textwidth]{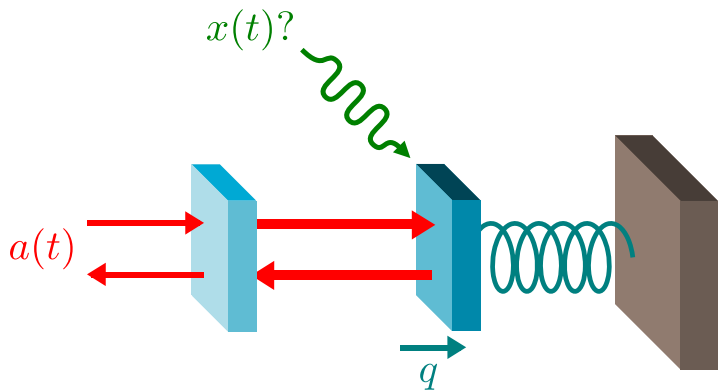}}
\caption{Schematic of a quantum optomechanical force sensor.}
\label{optomech}
\end{figure}

For a more complex example, consider the estimation of a force $x(t)$,
$t\in[t_0,t_f]$, on a quantum moving mirror via continuous optical
measurements, as illustrated in Fig.~\ref{optomech}.  For simplicity,
assume that any optical cavity dynamics can be adiabatically
eliminated \cite{chen2013}. Let $(q,p)$ be the mechanical position and
momentum operators, and $a(t)$ be the annihilation operator for the
one-dimensional optical field.  Suppose
\begin{eqnarray}
U_x = U[x(t)] = \mathcal T \exp\Bk{\frac{1}{i\hbar}\int_{t_0}^{t_f} dt H(t,x(t))},
\end{eqnarray}
with a Hamiltonian given by
\begin{eqnarray}
H(t,x(t)) = H_B(t,q,p,x(t)) -2\hbar Mk q I(t),
\end{eqnarray}
where $\mathcal T$ is the time-ordering superoperator, $H_B$ is the
mechanical Hamiltonian, $M$ is the effective number of optical
reflections by the mirror, $k$ is the optical wavenumber, and $I(t)
\equiv a^\dagger(t)a(t)$ is the photon flux.

In an optomechanics experiment, the mechanical oscillator is measured
only through the optical field, so one can take the mechanical Hilbert
space to be part of the inaccessible Hilbert space $\mathcal H_B$ and
replace $U_x$ by $ U_B^\dagger U_x$ with any $U_B$, according to
Sec.~\ref{purify}. Let
\begin{eqnarray}
U_B'[t_f,x(t)] \equiv \mathcal T\exp\Bk{\frac{1}{i\hbar}\int_{t_0}^{t_f} dt H_B(t,q,p,x(t))},
\\
U_I[x(t)] \equiv U_B'^\dagger[t_f,x(t)]U[x(t)],
\end{eqnarray}
which can be calculated using the interaction picture. The result is
\begin{eqnarray}
U_I[x(t)] = \mathcal T\exp\BK{2iMk\int_{t_0}^{t_f} dt q_I[t,x(t')] I(t)},
\\
q_I[t,x(t')] \equiv U_B'^\dagger[t,x(t')] q U_B'[t,x(t')].
\end{eqnarray}
If the mechanical dynamics is linear, $q_I$ can be
expressed in terms of the mechanical impulse-response function
$h(t,t')$ as 
\begin{eqnarray}
q_I[t,x(t')] = q_0(q,p,t) + \int_{t_0}^{t_f} dt'h(t,t')x(t'),
\end{eqnarray}
where $q_0$ is the transient solution.  To account for optical loss,
the techniques presented in the previous sections can be used, despite
the presence of an operator $q_0$ in the phase shift.  With the two
hypotheses given by $x(t) = x^{(0)}(t)$ and $x(t) = x^{(1)}(t)$, the
fidelity is
\begin{eqnarray}
F = \abs{\bra{\Psi}U_{AB}^\dagger U_I^\dagger[x^{(1)}(t)] 
U_B^\dagger U_I[x^{(0)}(t)] U_{AB}\ket{\Psi}}^2,
\label{optomech_F}
\end{eqnarray}
where $U_{AB}$ is given by (\ref{UAB}) and $U_B$ is given by
(\ref{UB}). As $U_B$ commutes with $U_I$,
(\ref{optomech_F}) can be rewritten as
\begin{eqnarray}
F = \abs{\bra{\Psi}U_{AB}^\dagger 
U_B^\dagger U_A U_{AB}\ket{\Psi}}^2,
\label{optomech_F2}
\\
U_A = U_I^\dagger[x^{(1)}(t)] U_I[x^{(0)}(t)]
=\exp\Bk{i\int_{t_0}^{t_f} dt \phi(t)I(t)},
\\
\phi(t) = 2Mk\int_{t_0}^{t_f} dt' h(t,t')\Bk{x^{(0)}(t')-x^{(1)}(t')}.
\end{eqnarray}
(\ref{optomech_F2}) is now identical to (\ref{ztransform2}),
the fidelity expression for optical phase waveform detection and
estimation.  Using $U_B'$, the mechanical Hilbert space has been
removed from the model, and the problem has been transformed to the
problem of sensing of a classical phase shift $\phi(t)$. The results
derived in the preceding sections can then be applied to this quantum
optomechanical sensing model.

\section{Relevance to quantum optics experiments}
The theoretical results presented here are especially relevant to the
experiments reported in \cite{wheatley,yonezawa,iwasawa}.  The
experiment in \cite{iwasawa}, in particular, applies a stochastic
force on a classical mirror probed by a continuous-wave optical beam
in coherent or phase-squeezed states. The waveform of interest $x(t)$
can then be the mirror position, momentum, or the force.  The phase
shift $\phi(t)$ is given by (\ref{functional}), which is a linear
functional of $x(t)$ with impulse-response function $g(t)$, and
measured in the experiment by a homodyne phase-locked loop, followed
by smoothing of the data
\cite{vantrees,tsl2008,tsl2009,smooth,smooth_pra1,smooth_pra2,petersen,gammelmark2013}.
The force, for example, is a realization of the Ornstein-Uhlenbeck
process, which has a prior power spectral density in the form of
\begin{eqnarray}
S_x(\omega) = \frac{\mu}{\omega^2+\nu^2}.
\end{eqnarray}
The prior Fisher information and the QCRB in spectral form becomes
\begin{eqnarray}
J_x^{(C)}(\omega) = \frac{1}{S_x(\omega)},
\\
\mathbb E [\tilde x(t)-x(t)]^2
\ge \intall \frac{d\omega}{2\pi}
\frac{1}{|g(\omega)|^2\min_\lambda J_\phi^{(Q)}(\omega) + 1/S_x(\omega)}.
\end{eqnarray}
The smoothing error, on the other hand, is
(Sec.~6.2.3 in \cite{vantrees})
\begin{eqnarray}
\mathbb E [\tilde x(t)-x(t)]^2
= \intall \frac{d\omega}{2\pi}
\frac{1}{|g(\omega)|^2/S_{\zeta}(\omega) + 1/S_x(\omega)},
\end{eqnarray}
where $S_\zeta(\omega)$ is the power spectral density of the homodyne
measurement noise. In \cite{iwasawa}, the experimental results
are compared with the QCRBs in terms of a QFI given by
\begin{eqnarray}
J_\phi^{(Q)}(\omega) = 4S_{\Delta I}(\omega).
\end{eqnarray}
The attainment of the bounds with this QFI requires a
minimum-uncertainty optical state, perfect phase-locking, and perfect
quantum efficiency, such that $S_\zeta(\omega)S_{\Delta I}(\omega) =
1/4$. One expects that the lower QFI given by (\ref{min_J}), taking
into account the imperfect quantum efficiency of the setup
($\eta\approx 87\%$), will make the QCRB even closer to the
experimental results, demonstrating the near-optimality of the
experimental techniques in the presence of loss.

It is intriguing to see from Sec.~\ref{sec_optomech} that the bound
remains valid even if the mirror is described by a quantum
model. Achieving the bound for a quantum mirror requires measurement
backaction noise in the output to be negligible relative to the
quantum-limited optical measurement noise. This may require quantum
noise cancellation \cite{klmtv,qnc,qmfs}.

The same setups in \cite{wheatley,yonezawa,iwasawa} may also be
used for the waveform detection experiment proposed in
Sec.~\ref{waveform_detect}. For coherent states and a known $\phi(t)$,
the measurement techniques demonstrated in
\cite{cook,wittmann,tsujino} may be generalized to attain the
bounds in Sec.~\ref{waveform_detect}. If $\phi(t)$ is stochastic, a
Kennedy receiver that nulls the field in the absence of phase
modulation should be able to achieve the optimal error exponent
\cite{tsang_nair}.  It remains an open question how optimal
measurements for phase-squeezed states can be implemented, but in
theory homodyne measurements should have an error exponent on the same
order as the fundamental limits.

Gravitational-wave detectors can nowadays operate at or below
shot-noise limits at certain frequencies
\cite{ligo2011,schnabel,chen2013}. This means that the bounds derived
here should be relevant if a gravitational wave falls within the
quantum-limited frequency bands. A detailed treatment, however, is
beyond the scope of this paper.

\section{Conclusion}
I have shown that tighter quantum limits can be derived for open
sensing systems by judicious purification.  For optomechanical force
sensing, the detection and estimation error bounds here should be
approachable using the quantum optics technology demonstrated in
\cite{cook,wittmann,tsujino,wheatley,yonezawa,iwasawa} and more
realistic to achieve than the bounds in \cite{twc,tsang_nair}.

\section*{Acknowledgments}
Discussions with R.~Nair, H.~Yonezawa, H.~Wiseman, and C.~Caves are
gratefully acknowledged. This work is supported by the Singapore
National Research Foundation under NRF Grant No.~NRF-NRFF2011-07.

\appendix

\section{\label{order}Order of loss and phase modulation}
Here I prove that the reduced state $\trace_B \rho_x$ is the same
regardless of the order of the optical loss $U_{AB}$ and the phase
modulation $U_A$, viz.,
\begin{lemma}
\begin{eqnarray}
\trace_B\bk{U_A U_{AB} \rho U_{AB}^\dagger U_A^\dagger}
= \trace_B\bk{U_{AB} U_A \rho U_A^\dagger U_{AB}^\dagger}
\end{eqnarray}
if $\rho = \rho_A \otimes \rho_B$ and $\rho_B$ is a thermal state.
\end{lemma}
\begin{proof}
  Here I consider one mode in $\mathcal H_A$ and one mode
  in $\mathcal H_B$; generalization to the multimode case is
  straightforward. Suppose
\begin{eqnarray}
U_A &= \exp(i \phi a^\dagger a ),
\\
U_{AB} &= \exp[i\kappa(a^\dagger b + a b^\dagger)],
\\
U_B &= \exp(i \phi b^\dagger b ),
\end{eqnarray}
where $[a,a^\dagger] = [b,b^\dagger]=1$. Then
\begin{eqnarray}
U_A U_{AB} U_A^\dagger 
= \exp[i\kappa(e^{i\phi}a^\dagger b + e^{-i\phi} a b^\dagger)]
= U_B^\dagger U_{AB} U_B,
\\
U_A U_{AB} = U_B^\dagger U_{AB} U_B U_A.
\end{eqnarray}
Since $U_B \rho_B U_B^\dagger = \rho_B$ for a thermal state,
\begin{eqnarray}
\trace_B\bk{U_A U_{AB} \rho U_{AB}^\dagger U_A^\dagger}
&= \trace_B\bk{U_B^\dagger U_{AB} U_B U_A \rho U_A^\dagger U_B^\dagger U_{AB}^\dagger U_B}
\\
&=\trace_B\bk{U_B^\dagger U_{AB} U_A \rho U_A^\dagger U_{AB}^\dagger U_B}
\\
&=\trace_B\bk{ U_{AB} U_A \rho U_A^\dagger U_{AB}^\dagger}.
\end{eqnarray}
\end{proof}

With the concatenation property of thermal-noise channels
\cite{giovannetti2004}, any optical loss with thermal noise at any
stage of a phase modulation experiment can be modeled by a single beam
splitter before or after the modulation.


\section{\label{su2}$SU(2)$ algebra}
Consider again the single-mode case with
\begin{eqnarray}
U_B = \exp(i\theta b^\dagger b).
\end{eqnarray}
First, compute the following quantity using the Heisenberg picture:
\begin{eqnarray}
U_{AB}^\dagger U_B^\dagger U_A U_{AB} = \exp(ig),
\\
g = \phi a'^\dagger a' - \theta b'^\dagger b',
\\
a' = \cos \kappa a + i\sin\kappa b,
\\
b' = \cos \kappa b + i\sin\kappa a.
\end{eqnarray}
This gives
\begin{eqnarray}
g &= \mu a^\dagger a + \nu b^\dagger b +i\gamma(a^\dagger b - a b^\dagger),
\\
\mu &\equiv \phi \cos^2\kappa-\theta\sin^2\kappa,
\\
\nu &\equiv \phi \sin^2\kappa - \theta\cos^2\kappa,
\\
\gamma &\equiv (\phi+\theta)\sin\kappa\cos\kappa.
\end{eqnarray}
Next, define $SU(2)$ operators as
\begin{eqnarray}
J_- &\equiv a^\dagger b,
\\
J_+ &\equiv a b^\dagger,
\\
J_3 &\equiv \frac{1}{2}\bk{b^\dagger b - a^\dagger a},
\\
J &\equiv \frac{1}{2}\bk{b^\dagger b + a^\dagger a},
\end{eqnarray}
where $J$ commutes with the rest of the operators.  In terms of the
redefined operators,
\begin{eqnarray}
\exp(ig) = \exp\Bk{i(\mu+\nu)J}
\exp\Bk{i(\nu-\mu)J_3 -\gamma J_- + \gamma J_+}.
\end{eqnarray}
The following theorem is useful:
\begin{theorem}[$SU(2)$ disentangling theorem]
  Given $J_\pm$ and $J_3$ that obey the commutation relations
\begin{eqnarray}
\Bk{J_+,J_-} = 2J_3,
\qquad
\Bk{J_3,J_\pm} = \pm J_\pm,
\label{spin}
\end{eqnarray}
the following identity holds:
\begin{eqnarray}
\exp(i\lambda_+ J_++i\lambda_- J_- + i\lambda_3 J_3)
= \exp(i\Lambda_+ J_+) \Lambda_3^{J_3}\exp(i\Lambda_-J_-),
\end{eqnarray}
where
\begin{eqnarray}
\Lambda_\pm &\equiv \frac{2\lambda_\pm \sin(\xi/2)}
{\xi\cos(\xi/2)-i\lambda_3\sin(\xi/2)},
\\
\Lambda_3 &\equiv \Bk{\cos(\xi/2)-i(\lambda_3/\xi)\sin(\xi/2)}^{-2},
\\
\xi &\equiv \sqrt{\lambda_3^2 + 4\lambda_+\lambda_-}.
\end{eqnarray}
\end{theorem}
\begin{proof}
See, for example, Chap.~7 in \cite{gilmore}.
\end{proof}
For the case of interest here,
\begin{eqnarray}
\lambda_3 &= \nu-\mu,
\qquad
\lambda_+ = -i\gamma,
\qquad
\lambda_- = i\gamma,
\\
\xi &= \sqrt{(\nu-\mu)^2+4\gamma^2}= \phi+\theta,
\\
\Lambda_3 &= \Bk{\cos \frac{\phi+\theta}{2}
- i(1-2\eta)\sin \frac{\phi+\theta}{2}}^{-2}.
\end{eqnarray}
The disentangling theorem is useful because $\exp(i\Lambda_-J_-)\ket{0}_B= \ket{0}_B$
and ${_B}\bra{0}\exp(i\Lambda_+J_+) = {_B}\bra{0}$:
\begin{eqnarray}
{_B}\bra{0}\exp(ig)\ket{0}_B
&= {_B}\bra{0}e^{i(\mu+\nu)J}
e^{i\Lambda_+ J_+}\Lambda_3^{J_3} e^{i\Lambda_- J_-}\ket{0}_B
\\
&= e^{i(\mu+\nu)a^\dagger a/2} \Lambda_3^{-a^\dagger a/2}
\\
&= z^{a^\dagger a},
\end{eqnarray}
where
\begin{eqnarray}
z \equiv \eta e^{i\phi} + (1-\eta)e^{-i\theta}. 
\end{eqnarray}

As an example, consider the fidelity for a coherent state
$\ket{\alpha}$:
\begin{eqnarray}
F = \abs{\bra{\alpha}z^{a^\dagger a}\ket{\alpha}}^2= 
\abs{\sum_{n=0}^\infty C_n z^n}^2,
\end{eqnarray}
where $C_n$ is the Poisson distribution with mean
$|\alpha|^2$. $\sum_n C_n z^n$ is known as the $z$-transform in
engineering and the generating function in statistics \cite{gardiner}.
It becomes the Fourier transform, also known as the characteristic
function in statistics, when $\eta = 1$. For the Poisson distribution,
\begin{eqnarray}
F &= \abs{\exp\Bk{|\alpha|^2(z-1)}}^2 
\\
&= \exp\BK{2|\alpha|^2\Bk{\eta\cos\phi+(1-\eta)\cos\theta-1}}.
\end{eqnarray}
To maximize $F$ and obtain the tightest lower bounds, one should choose $\cos\theta = 1$,
leading to (\ref{Fcoh}).

Generalization to the multimode case is straightforward. For
continuous optical fields, $a(t)$ can be first discretized in time as
$a(t_j) \approx \sqrt{\delta t} a_j$ with $[a_j,a_k^\dagger] =
\delta_{jk}$ before applying the multimode result and taking the
continuous limit. For example, a multimode coherent state with mean
photon flux $\avg{I(t)}$ can be written as a tensor product of
coherent states, each with a duration of $\delta t$ and mean number
$|\alpha_j|^2 = \avg{I(t_j)}\delta t$.  The collective fidelity is
then
\begin{eqnarray}
F &= \prod_{j}\exp\Bk{-4\eta \avg{I(t_j)}\delta t\sin^2\frac{\phi_j}{2}}
\\
&\to \exp\Bk{-4\eta\int dt\avg{I(t)}\sin^2\frac{\phi(t)}{2}},
\end{eqnarray}
which is (\ref{Fcoh2}).

\section{\label{qfi}Quantum Fisher information matrix}
Consider the multimode case with annihilation operators $a_j$ and
$b_j$. Let
\begin{eqnarray}
U_A = \exp\bk{i\sum_j \phi_j a_j^\dagger a_j},
\\
U_{AB} = \exp\Bk{i\kappa \sum_j \bk{a_j^\dagger b_j + a_j b_j^\dagger}},
\\
U_B = \exp\bk{i\sum_j \theta_j b_j^\dagger b_j},
\\
\theta_j = \sum_k \lambda_{jk}\phi_k.
\end{eqnarray}
The QFI matrix $\mathcal J^{(Q)}(\phi)$ can be computed by considering
the fidelity for small $\phi_j$ and $F\approx 1$ \cite{hayashi,paris}:
\begin{eqnarray}
F = \abs{\bra{\Psi}U_{AB}^\dagger  U_B^\dagger U_A U_{AB}\ket{\Psi}}^2
= 1-\frac{1}{4} \sum_{j,k} \mathcal J_{jk}^{(Q)}(\phi)\phi_j \phi_k + O(||\phi||^4).
\label{Fexpand}
\end{eqnarray}
This also shows why $F$ is more difficult to calculate than $\mathcal
J^{(Q)}(\phi)$ in general, as $\mathcal J^{(Q)}(\phi)$ is just a second-order term in $F$.
Write the fidelity as
\begin{eqnarray}
F &\equiv \abs{\bra{\Psi}\exp\bk{i\sum_j g_j}\ket{\Psi}}^2,
\\
g_j &\equiv \mu_j a_j^\dagger a_j + \nu_j b_j^\dagger b_j 
+i\gamma_j(a_j^\dagger b_j - a_j b_j^\dagger),
\\
\mu_j &= \phi_j\cos^2\kappa-\theta_j\sin^2\kappa,
\\
\nu_j &= \phi_j\sin^2\kappa - \theta_j\cos^2\kappa,
\\
\gamma_j &= (\phi_j+\theta_j)\sin\kappa\cos\kappa.
\end{eqnarray}
Since $g$ is a linear function of $\phi$, one can first expand $F$ in the
leading order of $g$:
\begin{eqnarray}
F &\approx \abs{\bra{\Psi} 1 + i\sum_j g_j-\frac{1}{2}\sum_{j,k}g_j g_k\ket{\Psi}}^2
\\
&= \bk{1-\frac{1}{2}\sum_{j,k}\bra\Psi g_j g_k\ket\Psi}^2 + \bk{\sum_j \bra\Psi g_j\ket\Psi}^2
\\
&\approx
1-\sum_{j,k}\bra\Psi g_j g_k\ket\Psi + \sum_{j,k} \bra\Psi g_j\ket\Psi\bra\Psi g_k\ket\Psi
\\
&=1-\sum_{j,k}\bra\Psi\Delta g_j\Delta g_k\ket\Psi,
\qquad
\Delta g_j \equiv g_j - \bra{\Psi}g_j\ket{\Psi},
\label{gvar}
\end{eqnarray}
and obtain $\mathcal J^{(Q)}(\phi)$ by computing $\bra\Psi\Delta
g_j\Delta g_k\ket\Psi$ and comparing (\ref{Fexpand}) and
(\ref{gvar}). After some algebra,
\begin{eqnarray}
\frac{\mathcal J_{jk}^{(Q)}(\phi)}{4} &= \sum_{l,m}(\delta_{lj}\cos^2\kappa-
\lambda_{lj}\sin^2\kappa) 
(\delta_{mk}\cos^2\kappa-\lambda_{mk}\sin^2\kappa) 
\Avg{\Delta n_l\Delta n_m}
\nonumber\\&\quad
+ \sum_{l,m}\sin^2\kappa\cos^2\kappa
(\delta_{lj}+\lambda_{lj})(\delta_{mk}+\lambda_{mk})
 \Avg{n_l}\delta_{lm},
\end{eqnarray}
where $n_j \equiv a_j^\dagger a_j$. Since $\mathcal J^{(Q)}(\phi)$
does not depend on $\phi$, the Bayes QFI $J^{(Q)}$ is equal to
$\mathcal J^{(Q)}$. In the continuous-time limit with $t_j =
t_0+j\delta t$, $\delta t\to 0$,
\begin{eqnarray}
\frac{1}{\sqrt{\delta t}} a_j \to a(t_j),
\qquad
\frac{1}{\delta t} n_j \to I(t_j),
\qquad
\frac{\delta_{jk}}{\delta t} \to \delta(t_j-t_k),
\nonumber\\
\frac{\lambda_{jk}}{\delta t} \to \lambda(t_j-t_k),
\qquad
\frac{J_{jk}^{(Q)}}{\delta t^2} \to J^{(Q)}(t_j,t_k),
\end{eqnarray}
and (\ref{JQphi}) in the main text is obtained.

\section*{References}
\bibliographystyle{hunsrt}
\bibliography{noisy_metrology3_njp2}

\end{document}